\documentclass{amsproc}
\usepackage{cite}
\usepackage{amsmath,amsfonts,amssymb,amsthm,amscd,mathrsfs,mathdots,bbm,color}

\usepackage{graphicx}
\usepackage[linecolor=white,backgroundcolor=orange,bordercolor=red]{todonotes}

\usepackage{ stmaryrd }

\def\R{{\mathbb R}}
\def\C{{\mathbb C}}
\def\N{{\mathbb N}}

\def\ts{\otimes}

\def\EE{{\mathbb E}}

\def\l{\lambda}

\def\Tr{\mathrm{Tr}}
\def\tr{\mathrm{tr}}

\def\Ton{\mathrm{T}}

\def\Hm{\mathrm{H}^{\uparrow}}
\def\Wm{\mathcal{W}^{\uparrow }}
\def\Fm{\mathcal{F}^{\uparrow}}
\def\l{\lambda}
\def\b{\beta}

\def\Hsm{\mathrm{H}^{\uparrow \hspace{-0,05 cm}\uparrow}}
\def\Wsm{\mathcal{W}^{\uparrow \hspace{-0,05 cm}\uparrow}}
\def\Fsm{\mathcal{F}^{\uparrow \hspace{-0,05 cm}\uparrow}}

\newtheorem{thm}{Theorem}
\newtheorem{conj}{Conjecture}
\newtheorem{cor}{Corollary}
\newtheorem{prop}{Proposition}
\newtheorem{lem}{Lemma}
\newtheorem{rmk}{Remark}
\newtheorem{ex}{Example}
\theoremstyle{definition}
\newtheorem*{notation}{Notation}

\newcommand{\be}{\begin{equation}}
\newcommand{\ee}{\end{equation}}

\newcommand{\ben}{\begin{equation*}}
\newcommand{\een}{\end{equation*}}

\title[Integer moments of complex Wishart matrices and Hurwitz numbers]{Integer moments of complex Wishart matrices\\ and Hurwitz numbers
}

\author[F. D. Cunden]{Fabio Deelan Cunden}
\address{School of Mathematics and Statistics, University College Dublin, Belfield, Dublin 4, Ireland}
\email{fabio.cunden@ucd.ie}
\author[A. Dahlqvist]{Antoine Dahlqvist}
\address{School of Mathematics and Statistics, University College Dublin, Belfield, Dublin 4, Ireland}
\email{antoine.dahlqvist@ucd.ie}
\author[N. O'Connell]{Neil O'Connell}
\address{School of Mathematics and Statistics, University College Dublin, Belfield, Dublin 4, Ireland}
\email{neil.oconnell@ucd.ie}
\thanks{Research of FDC, AD and NO'C supported by ERC Advanced Grant 669306. Research of FDC partially supported by Gruppo Nazionale di Fisica Matematica  GNFM-INdAM}

\begin{document}
  \maketitle
  
\begin{abstract} We give formulae for the cumulants of  complex Wishart (LUE) and inverse Wishart matrices (inverse LUE). Their large-$N$ expansions are generating functions of double (strictly and weakly) monotone Hurwitz numbers which count constrained factorisations in the symmetric group. The two expansions can be compared and combined with a duality relation proved in  [F. D. Cunden, F. Mezzadri, N. O'Connell and N. J. Simm,  arXiv:1805.08760] to obtain:  i) a combinatorial proof of the reflection formula between moments of LUE and inverse LUE at genus zero and, ii) a new functional relation between the generating functions of monotone and strictly monotone Hurwitz numbers. The main result  resolves the integrality conjecture formulated in [F. D. Cunden, F. Mezzadri, N. J. Simm and P. Vivo,  J.~Phys.~A 49 (2016)] on the time-delay cumulants in quantum chaotic transport. The precise combinatorial description of the cumulants given here may cast new light on the concordance between random matrix and semiclassical theories.
\end{abstract}

\section{Introduction and results}
\subsection{Time-delay matrix and an integrality conjecture}
\par
Random matrices have been used to model a variety of scattering phenomena in complex systems including heavy nuclei, disordered mesoscopic 
conductors, and chaotic quantum billiards. See, e.g.,~\cite{Been,GUHR1998,Weiden,Schomerus}. The time-dependent aspects of a scattering process are usually 
described by the time-delay (or Wigner-Smith) matrix $Q$. Its eigenvalues $\tau_j$ are called \emph{proper delay times} 
and can be thought as the time spent by an incident wave in the scattering region at a  propagating mode (or open channel) $j=1,\dots,N$. See~ 
\cite{T16} for a modern introduction.
\par
A statistical approach to the time-delay based on random matrices was developed in the 1990s, see~\cite{LSSS,FS,SSS,GMB}. For ballistic quantum dots with 
perfect coupling (a physical realisation of chaotic quantum billiards), Brouwer, Frahm, and Beenakker~\cite{BFB97,BFB99} argued that the inverses of the 
proper delay times $\l_j=(N\tau_j)^{-1}$ are distributed according to the Laguerre ensemble of random matrix theory
\begin{align}
p(d\l_1,\dots,d\l_N)&=c_{N,\b}\prod_{i<j}|\l_i-\l_j|^{\beta}\prod_{k}\l_k^{\b N/2}e^{-\b N\l_k/2}\chi_{\R_+}(\l_k)d\l_k,\label{eq:Lag}
\end{align}
where $\beta\in\{1,2,4\}$ indicates orthogonal, unitary, or symplectic symmetry, respectively, and $c_{N,\b}$ is a normalisation constant. 
This provided a route to apply various techniques from random matrix theory for the calculations of expectation values, typical fluctuations and tails of 
the distributions of the time-delay moments $\tr Q^k$. See~\cite{MMG,SFS,BK,KSS,TM,C,CMSV,CMSV2,Novaes1,Novaes11,MS1,MS2,MS3}.
\begin{notation}
$\Tr$ denotes the non-normalised trace on $\mathcal{M}_N(\C)$, and $\tr= \frac 1 N \Tr.$  For $n\in\N,$ we set $[n]=\{1,\dots,n\}$, and $\mathcal{P}(n)$ is 
the set of partitions of $[n]$. If $(Y_1\dots,Y_\ell)$ are random variables (not necessarily distinct) on the same probability space with finite moments, their $\ell$th cumulant (or connected average) is defined according to the formula $C_\ell(Y_1,\dots,Y_\ell)=\sum_{\pi\in\mathcal{P}(\ell)}(|
\pi|-1)!(-1)^{|\pi|-1}\prod_{B\in\pi}\EE\prod_{i\in B}Y_i$.
\end{notation}
\par
The joint law~\eqref{eq:Lag} of the eigenvalues of $W=(NQ)^{-1}$ defines a $\beta$-ensemble ($\beta>0$) with a strictly convex potential. This case belongs to the class of  one-cut, off-critical ensembles, for which Borot and Guionnet~\cite{BG} proved the existence of asymptotic $1/N$-expansions determined by recursive relations known as `loop equations'.
For instance, the generating series of the cumulants (also called `correlators')
 \be
G_{\ell,\beta}(z_1,\dots,z_l)=C_\ell\left(\tr\frac{1}{z_1-W},\dots,\tr\frac{1}{z_l-W}\right)
\ee
admit large-$N$ asymptotic expansions of the form
\be
G_{\ell,\b}(z_1,\dots,z_\ell)=\frac{1}{\left(\beta N^2\right)^{n-1}}\sum_{g\geq0}N^{-g}G_{\ell,\b}^{\{g\}}(z_1,\dots,z_\ell),
\label{eq:topol}
\ee
where $G_{\ell,\b}^{\{g\}}$ has a very simple dependence in $\b$
\be
G_{\ell,\b}^{\{g\}}(z_1,\dots,z_\ell)=\sum_{k=0}^{\lfloor g/2\rfloor}\b^{-k}\left(\frac{1}{2}-\frac{1}{\b}\right)^{g-2k}G_{\ell}^{\{k;g-2k\}}(z_1,\dots,z_\ell).
\ee
(See~\cite{BG} for details.) The coefficients $G_{\ell}^{\{k;g-2k\}}$ can be computed recursively using the Chekhov-Eynard topological recursion\cite{CE}. It is easy to check that, 
when $\b=2$,  $G_{\ell,2}^{\{g\}}=0$ if $g$ is odd, and~\eqref{eq:topol} is an expansion in powers of $1/N^2$.
\par
In \cite{CMSV}, using methods devised by Ambj\o rn, Chekhov, Kristjansen, and Makeenko\cite{ACKM}, 
the explicit form of the leading order $G^{\{0\}}_{\ell,\b}(z_1,\dots,z_\ell)$, and the large-$N$ limit of the cumulants
\be
\lim_{N\to\infty}\left(\b N^2\right)^{\ell-1}C_\ell\left(\tr W^{-\mu_1},\dots,\tr W^{-\mu_\ell}\right)=c_0(\mu_1,\dots,\mu_\ell)
\label{eq:limit}
\ee 
were analysed. (The limit does not depend on $\b$.) Extensive computations of some families of 
$c_0(\mu_1,\dots,\mu_\ell)$'s led the authors to the following integrality conjecture.
\begin{conj}[\!\!\cite{CMSV}] 
\label{conj:integers} For all $\ell\geq1$ and $(\mu_1,\dots,\mu_\ell)\in\N^{\ell}$, 
\ben
c_0(\mu_1,\dots,\mu_\ell)\in\N.
\een
\end{conj}
\par
The present work started as an attempt to prove the conjecture. In this paper we provide an explicit formula for the $1/N$-expansion of the cumulants 
\be
C_\ell\left(\tr W^{-\mu_1},\dots,\tr W^{-\mu_\ell}\right)=\frac{1}{\left(2 N^2\right)^{\ell-1}}\sum_{g\geq0}N^{-g}c_g(\mu_1,\dots,\mu_\ell)
\label{eq:asymp}
\ee 
when $\beta=2$. The result not only resolves Conjecture~\ref{conj:integers}, but shows that \emph{the full $1/N$-expansion has positive integer coefficients} (i.e., $c_g(\mu_1,\dots,\mu_\ell)\in\N$) whose combinatorial interpretation we describe completely in 
terms of constrained factorisations in the symmetric group. In fact, the large-$N$ asymptotics~\eqref{eq:asymp} is a `genus' expansion. 

\subsection{Complex Wishart matrices and the Laguerre unitary ensemble} For any real number $M> N-1$, consider the following probability measure supported on the cone of  positive definite $N\times N$ complex Hermitian matrices
\be
 \gamma(dX)=\frac{N^{NM}}{\pi^{N(N-1)/2}\prod_{j=0}^{N-1} \Gamma(M- j)} \left(\det X\right)^{ M-N} \exp\left({- N\Tr X }\right) dX.
\ee
A random matrix $W$ distributed according to the above measure is a \emph{complex Wishart matrix} with parameter $M$. It is also 
quite common to use the parameters $c=M/N$, or $\alpha=M-N$. The eigenvalues of $W$ (we drop the dependence on $N$ and $c$ for notational 
convenience) are distributed according to
\begin{align*}
p(d\l_1,\dots,d\l_N)&=c_{N}\prod_{i<j}|\l_i-\l_j|^{2}\prod_{k}\l_k^{M-N}e^{-N\l_k}\chi_{\R_+}(\l_k)d\l_k\\
\label{eq:Lag2}
&c_{N}^{-1}=\frac{N!}{N^{MN}}\prod_{j=1}^{N}\Gamma\left(\alpha+j\right)\Gamma\left(j\right).
\end{align*}
This is the Laguerre Unitary Ensemble (LUE for short) of random matrix theory.  When $M$ is an integer, there is the equality in law $W=N^{-1}XX^{\dagger}$, where $X$ is a $N \times M$ random matrix with independent standard Gaussian entries~\cite{Muirhead}. When $\beta=2$, Eq.~\eqref{eq:Lag} is of this type for the particular 
choice $M=2N$ (or $c=2$, $\alpha=N$). 

\subsection{Statement of results}

\begin{notation} When $\sigma$ is a permutation, an integer partition or a set partition, we denote by $\#\sigma$ its number of cycles (resp. blocks). 
For a random matrix $X$ of size $N$, with coefficients having joint moments  of  homogeneity $n\in \N$, we shall denote  for any integer partition  $
\mu=(\mu_1,\ldots,\mu_\ell) \vdash n, $  the scaled cumulant
 \begin{equation}
  \mathcal{C}_X(\mu)=  \frac{|\mu|!}{z_\mu} N^{2(\#\mu-1)} C_{\#\mu}(\tr(X^{\mu_1}), \ldots, \tr(X^{\mu_l})), \label{ScaledCumulant}
 \end{equation}
 where $|\mu|=n$, $\#\mu=\ell$ and $z_\mu=\prod_{i\geq1}m_i!i^{m_i}$ ($m_i$ being the number of parts of $\mu$ equal to $i$).
\end{notation}
\par
The main purpose of this paper is to explain that, for the LUE and inverse LUE, ~\eqref{ScaledCumulant} counts combinatorial quantities, related to 
factorisations in the symmetric group.
    \begin{thm}
    \label{MainTh} 
    Fix $n\in \N$, $n\geq1$, and $\mu\vdash n$. Then,  
  \begin{align}
  \mathcal{C}_{W^{-1}}(\mu)  &=   \sum_{g \ge 0 }    N^{-2g}\sum_{\nu\vdash n }   (c-1)^{-(n+ 2g-2+\#\mu+\#\nu)}   \Hm_g(\mu,\nu) &\text{for }&c> 1+
\frac{n}{N},\label{mainth:CumulantInverse}\\
  \mathcal{C}_W(\mu)& =   \sum_{g \ge 0 }    N^{-2g}\sum_{\nu\vdash n }  c^{n-(2g-2+\#\mu+\#\nu)}   \Hsm_g(\mu,\nu) &\text{for }&c> 
1-\frac{1}{N}.\label{propeq:CumulantWishartE}
    \end{align}
$\Hm_g(\mu,\nu) $ is the number of  tuples  $(\alpha, \tau_1,\ldots, \tau_{r}, \beta)$, where 
\begin{enumerate}
\item[($\mathrm{\romannumeral 1}$)] $r= \#\mu+\#\nu+2g-2$; 
\item[($\mathrm{\romannumeral 2}$)] $\alpha, \beta \in S_n$ are respectively  permutations  of type $\mu$ and $\nu$ and $\tau_1,\ldots,\tau_r$  are 
transpositions such that 
$$ \alpha \tau_1\ldots \tau_r= \beta;$$
\item[($\mathrm{\romannumeral 3}$)]  the group generated by $(\alpha, \tau_1,\ldots, \tau_r)$ acts transitively on $[n]$;
\item[($\mathrm{\romannumeral 4}$)]   $\tau_1,\ldots, \tau_r$ being written as $\tau_i= (a_i\, b_i)$ with $a_i<b_i,$   
\ben
b_1\le b_2\le \ldots \le  b_r. 
\een
\end{enumerate}
$\Hsm_g(\mu,\nu) $ is the number of  tuples  $(\alpha, \tau_1,\ldots, \tau_{r}, \beta) ,$  satisfying all the four conditions above but the last one, which 
is replaced by 
\begin{enumerate}
\item[($\mathrm{\romannumeral 4}$')]  $\tau_1,\ldots, \tau_r$ being written as $\tau_i= (a_i\, b_i)$ with $a_i<b_i,$  
\ben
b_1< b_2< \ldots <  b_r.
\een
\end{enumerate} 
   \end{thm} 
Note that the strict monotonicity condition ($\mathrm{\romannumeral 4}$') truncates the sum in $g$, and 
$\mathcal{C}_{W}(\mu)$ is a \emph{polynomial} in $1/N^2$ (this is well known). The series representation~\eqref{mainth:CumulantInverse} of the cumulants $\mathcal{C}_{W^{-1}}(\mu)$ is not asymptotic but \emph{convergent} for $N>n/(c-1)$.
\par
 The fact that $\mathcal{C}_{W^{-1}}(\mu)$ and $\mathcal{C}_{W}(\mu)$ can be written as sums over permutations is a consequence of the 
\emph{Schur-Weyl duality}, which applies to \emph{any} unitarily invariant ensemble. Explicit formulae for the coefficients in the sum are only known 
for special cases, e.g., GUE, CUE, and LUE. In fact, the expression~\eqref{propeq:CumulantWishartE} is folklore in the literature \cite{CMSS,HSS,LM}. The 
new result here is the explicit formula~\eqref{mainth:CumulantInverse} which shows that the class of `solvable' matrix ensembles includes the inverse 
LUE too.
\par
 The numbers $\Hm_g(\mu,\nu)$ (resp. $\Hsm_g(\mu,\nu)$) in the above Theorem are known as \emph{monotone} (resp.  \emph{strictly monotone}) 
\emph{double Hurwitz numbers}, a special class of \emph{Hurwitz numbers}.  The latter count factorisations without the condition ($
\mathrm{\romannumeral 4}$) or ($\mathrm{\romannumeral 4}$')  which are in bijections with labeled connected ramified covering of the sphere of 
degree $n$, with  ramifications of type $\mu,\nu$ and $r$ simple ramifications, with a total space  defining a surface of genus $g.$ For a beautiful 
introduction see \cite{LandoZvonkine}.  The above statement can also be reformulated in terms of \emph{prefixes of minimal factorisations}, see 
Theorem \ref{MainThV2} below and, when $\mu$ has one block, in terms of \emph{parking functions}, see \cite{LFacto,BianeParking,StanleyParking}.  
When $\#\mu=n$, $\Hm_g((1,\ldots,1),\nu)$ is the number of  \emph{primitive factorisations} of any permutation of cycle type $\nu$ into $r$ transpositions, see~\cite{MN2,GM}.
 \par
 The main ingredients of the proof hinges on the combination of two results: i) a formula for the expectation of coefficients of inverse Wishart 
matrices found by Graczyk, Letac, and Massam~\cite{GLM} (in its reformulation in terms of Weingarten function due to Collins, Matsumoto, and 
Saad~\cite{CMS}), and ii) the expression of the Weingarten function in terms of Jucys-Murphy elements~\cite{Jucys} due to Novak~\cite{Novak}. 
The paper~\cite{GN2011} by Gupta and Nagar contains some hints on the existence of explicit formulae for the cumulants of the inverse LUE, and was instrumental in our study.

First expressions for asymptotics of the Weingarten function were examined in~\cite{CollinsPhD} using representation theory and then developed in 
\cite{CMSS,CM}, to study scaled cumulants of unitary invariant matrix ensembles, in terms of the poset of partitioned permutations. The introduction of 
monotone Hurwitz numbers for the study of the 
Harisch-Chandra-Itzykson-Zuber integrals and unitary invariant matrix models was initiated 
in~\cite{Novak,MN,GGPNHCIZ}, see also~\cite{GN,DDD,BGF} for recent studies of these observables thanks to topological recursion.
\par
An immediate application of the main Theorem for $c=2$ is the following corollary on the time-delay matrix.      
 \begin{cor} When $\beta=2$ (unitary symmetry), the large-$N$ expansions~\eqref{eq:asymp} of the cumulants of the time-delay matrix have positive 
integer coefficients. More precisely, $c_{2g+1}=0$, and
 \be
c_{2g}(\mu_1,\dots,\mu_\ell)=
 2^{\ell-1}\frac{z_{\mu}}{|\mu|!}\sum_{\nu\vdash|\mu|}\Hm_{g}(\mu,\nu)\in\N.
 \label{eq:TD}
 \ee
 (This implies, in particular, Conjecture~\ref{conj:integers}.)
 \end{cor}
 \begin{ex} Let $n=3$, $\mu=(1,1,1)$, and $g=0$. We outline the calculations of $\Hm_0((1,1,1),\nu)$ and $\Hsm_0((1,1,1),\nu)$. The integer partitions $\nu\vdash n$ are $\nu=(3),(2,1)$, and $(1,1,1)$.
 \begin{itemize}
 \item[$\nu=(3)$:] There are $\binom{3}{2}^2=9$ products of $r=2$ transpositions in $S_3$:
$$
\begin{array}{ccc}
(1\, 2)(1\, 3)  &(1\, 2)(2\, 3) &(2\, 3)(1\, 2) \\
(2\, 3)(1\, 3)  &(1\, 3)(2\, 3) &(1\, 3)(1\, 2)  \\
(1\, 2)(1\, 2)  &(1\, 3)(1\, 3)  &(2\, 3)(2\, 3)\\
\end{array}
$$
$6$ of them are transitive (the first two rows in the table above) and produce a cycle type $\mu=(3)$, 
but only the $4$ products in the upper-left corner are monotone, so $\Hm_0((1,1,1),(3))=4$. The number of strictly monotone products is $\Hsm_0((1,1,1),(3))=2$.
 \item[$\nu=(2,1)$:] There are $\binom{3}{2}^3=27$ products of $r=3$ transpositions, and $24$ of them are transitive and produce a cycle type $\nu=(2,1)$. Only $12$ products are monotone
$$
\begin{array}{cccc}
(1\, 2)(1\, 2)(1\, 3)  &(1\, 2)(1\, 3)(2\, 3) &(1\, 3)(1\, 3)(2\,3)  &(2\, 3)(1\, 3)(1\,3)\\
(1\, 2)(1\, 2)(2\, 3)  &(1\, 2)(2\, 3)(1\, 3) &(1\, 3)(2\, 3)(1\, 3) &(2\, 3)(1\, 3)(2\, 3) \\
(1\, 2)(1\, 3)(2\, 3) &(1\, 2)(2\, 3)(2\, 3) &(1\, 3)(2\, 3)(2\, 3) &(2\, 3)(2\, 3)(1\, 3) \\
\end{array}
$$
so $\Hm_0((1,1,1),(2,1))=12$, but none of them is strictly monotone, so $\Hsm_0((1,1,1),(2,1))=0$.
 \item[$\nu=(1,1,1)$:] Among the $\binom{3}{2}^4=81$ products of $r=4$ transpositions, only $8$ of them are transitive, produce a cycle type $\nu=(1,1,1)$ and are monotone, so $\Hm_0((1,1,1),(1,1,1))=8$, 
$$
\begin{array}{cc}
(1\, 2)(1\, 2)(1\, 3)(1\, 3)  &(1\, 2)(1\, 2)(2\, 3) (2\, 3)\\ 
(1\, 2)(1\, 3)(2\, 3)(1\, 3) &(1\, 2)(2\, 3)(1\, 3)(2\, 3) \\
(2\, 3)(2\, 3) (1\, 3)(1\, 3)  &(1\, 3)(1\, 3)(2\, 3) (2\, 3)\\
(1\, 3)(2\, 3) (1\, 3)(2\, 3) &(2\, 3)(1\, 3) (2\, 3)(1\, 3)  \\ 
\end{array}
$$
There are no strictly monotone products, so $\Hsm_0((1,1,1),(1,1,1))=0$.
 \end{itemize}
From the above calculations we can conclude that ($\frac{z_{(1,1,1)}}{|(1,1,1)|!}=1$)
 \begin{align*}
 \mathcal{C}_3(\tr W^{-1},\tr W^{-1},\tr W^{-1})&=\frac{1}{N^4}\left(\sum_{\nu\vdash 3}\frac{\Hm_0((1,1,1),\nu)}{(c-1)^{4+\#\nu}}+O(N^{-2})\right)\\
 &=\frac{1}{N^4}\left(\frac{4}{(c-1)^{5}}+\frac{12}{(c-1)^{6}}+\frac{8}{(c-1)^{7}}+O(N^{-2})\right)\\
 \mathcal{C}_3(\tr W^{1},\tr W^{1},\tr W^{1})&=\frac{1}{N^4}\left(\sum_{\nu\vdash 3}c^{2-\#\nu}\Hsm_0((1,1,1),\nu)+O(N^{-2})\right)\\
 &=\frac{1}{N^4}\left(2c+O(N^{-2})\right).
 \end{align*}
 These agree with known results~\cite{MS3,CMSV}.
 \end{ex}
 \begin{ex} 
 \label{ex:trace} We compute $\EE\tr W^{-1}$ and $\EE\tr W$. These cases correspond to the one-block $\mu=(1)$.  From the formulae~\eqref{mainth:CumulantInverse}-\eqref{propeq:CumulantWishartE}, we have
 \begin{align*}
   \EE\tr W^{-1}=\frac{z_{(1)}}{|(1)|!}\frac{1}{N^{2(\#(1)-1)}}\mathcal{C}_{W^{-1}}((1))&=\sum_{g\geq0}N^{-2g}\sum_{\nu\vdash 1}(c-1)^{-1-2g}\Hm_g((1),\nu)\\
 \EE\tr W=\frac{z_{(1)}}{|(1)|!}\frac{1}{N^{2(\#(1)-1)}}\mathcal{C}_{W}((1))&=\sum_{g\geq0}N^{-2g}\sum_{\nu\vdash 1}c^{1-2g}\Hsm_g((1),\nu)
 \end{align*}
Using $\Hm_g((1),(1))=\Hsm_g((1),(1))=\delta_{g0}$, we recover the well-known results~\cite{GLM,GN2011}
 \ben
 \EE\tr W^{-1}=(c-1)^{-1},\qquad
 \EE\tr W=c.
 \een
 \end{ex}
      \begin{ex} Set $c=2$ (or $\alpha=N$). We want to compute the second moment of the time-delay matrix $\EE\tr W^{-2}$, corresponding to the one-block partition $\mu=(2)$. 
From the definition~\eqref{ScaledCumulant}
 and formula~\eqref{mainth:CumulantInverse},
  \ben
 \EE\tr W^{-2}=\frac{z_{(2)}}{|(2)|!}\frac{1}{N^{2(\#(2)-1)}}\mathcal{C}_{W^{-1}}((2))=\sum_{g\geq0}N^{-2g}\sum_{\nu\vdash 2}\Hm_g((2),\nu),
 \een
where we used $n=|(2)|=2$, $\ell=\#(2)=1$, and $z_{(2)}=2$. There are two possibilities: $\nu=(2)$ and $\nu=(1,1)$. Therefore we must count the 
monotone solutions of the factorisation problems ($\mathrm{\romannumeral 2}$) in $S_2$
\ben
\begin{cases}
(1\,2)\tau_1\cdots\tau_{2g}=(1\,2) &\text{if $\nu=(2)$}\\
(1\,2)\tau_1\cdots\tau_{2g+1}=\mathrm{id} &\text{if $\nu=(1,1)$}
\end{cases}
\een
where we used condition ($\mathrm{\romannumeral 1}$). In $S_2$ there is only one transposition, $\tau=(1\,2)$. Therefore in both cases there is only 
one path of the form $(1\,2)\cdots(1\,2)$ (with $2g$ factors if $\nu=(2)$, and $2g+1$ factors if $\nu=(1,1)$), and this path is also connected and 
monotone (conditions ($\mathrm{\romannumeral 3}$) and ($\mathrm{\romannumeral 4}$)). Hence, $\Hm_g((2),(2))=\Hm_g((2),(1,1))=1$. Substituting 
in the formula, we get
\ben
 \EE\tr W^{-2}=\sum_{g\geq0}N^{-2g}\sum_{\nu\in\{(2),(1,1)\}}1=\frac{2}{1-N^{-2}}=\frac{2N^2}{N^2-1},
\een
in agreement with the known result~\cite[Appendix A]{CMSV2}. Note that $\EE\tr W^{1}=2$ (see Example~\ref{ex:trace}); c.f. the reciprocity formula~\eqref{eq:reciprocity_or} below.
     \end{ex}
      \begin{rmk}[Physical significance of Theorem~\ref{MainTh}] The random matrix theory approach to quantum chaos is believed to be equivalent to perturbative calculations based on semiclassical considerations. In the time-delay problem, the random matrix averages correspond to sums over \emph{pairs of correlated classical trajectories} connecting the leads (asymptotic waves) with the interior of the cavity (the scattering region). In fact, some hints in the formulation of Conjecture~\ref{conj:integers} came from the observation that the semiclassical calculations boil down to weighted enumeration of \emph{diagrams} recording only the topology of the trajectories.
 \par
The concordance between random matrix and semiclassical theories in open systems has been established recently by Berkolaiko and Kuipers~\cite{BK2,BK3,BK4} and Novaes~\cite{Novaes3} in the case of quantum transport (when the relevant matrix model is the CUE). They put the diagrammatic method of the semiclassical approximation on a rigorous footing, and recast the semiclassical evaluation of moments as a summation over factorisations of given permutations (implying that the contribution of a diagram is given by the unitary, or orthogonal, Weingarten function). 
\par
On the other hand, for the time-delay, the agreement between semiclassics and random matrices remains limited to the first eight moments~\cite{Novaes2}, and to the leading and several subleading orders in the $1/N$-expansion~\cite{KSS}. By Theorem~\ref{MainTh}, the coefficients in the $1/N$-expansion of the time-delay are positive integers, thus supporting the equivalence with the semiclassical diagrammatic rules. Moreover,~\eqref{eq:TD} provides an explicit formula for the cumulants as a sum over monotone factorisations of permutations (which are related to the Weingarten function). It may not be too much to hope that this result will stimulate further study of the semiclassical diagrams in the time-delay problem to establish the equivalence with random matrices to all orders in $1/N$.
  \end{rmk}
  
In the proof we shall first get a less symmetric version of Theorem \ref{MainTh}. 

\begin{thm}\label{MainThV2} For any permutation  $\alpha\in S_n$ with cycle type $\mu=(\mu_1,\ldots, \mu_\ell)\vdash n,$
\begin{equation}
N^{2 (\ell-1)}C_\ell(\tr(W^{-\mu_1}), \ldots, \tr(W^{-\mu_\ell}))= \sum_{r,d\ge 0} N^{-2d}  (c-1)^{-n-r}\# \Fm_{n,r,d}(\alpha),
\label{eq:1st}
\end{equation}
and 
\begin{equation}
N^{2 (\ell-1)}C_\ell(\tr(W^{\mu_1}), \ldots, \tr(W^{\mu_\ell}))= \sum_{r,d\ge 0} N^{-2d}  c^{n-r}\# \Fsm_{n,r,d}(\alpha)
\label{eq:2nd}
\end{equation}
where  $\Fm_{n,r,d}(\alpha)$ (resp.  $\Fsm_{n,r,d}(\alpha)$) is the set of transpositions tuples $(\tau_1,\ldots,\tau_r)$  where  $\tau_i=(a_i\,b_i)$ with $a_i<b_i$ for all 
$i$, such that
\begin{enumerate}
\item  $\# \alpha \tau_1\ldots \tau_r= \#\alpha +r-2d, $
\item  $\langle \alpha,\tau_1,\ldots, \tau_r \rangle $ acts transitively on $[n],$
\item   $b_1\le b_2\le \ldots \le b_r$ (resp. $b_1<b_2<\ldots <b_r$).
\end{enumerate}
\end{thm}
\begin{rmk}
Within the Cayley graph on $S_n$ generated by  all transpositions, the distance between two permutations $\alpha$ and $\beta$  is   $ d(\alpha,\beta)= |\alpha^{-1}\beta|,$ where for any $\sigma\in S_n,$ $|\sigma|=n-\#\sigma.$  Any element   $(\tau_1,\ldots, \tau_r)$ of  $\Fm_{n,r,d}$   and $\Fsm_{n,r,d}$ defines a path in $S_n$ with $r$ steps that starts at  $\alpha$ and  ends at $\beta=\alpha.\tau_1\ldots \tau_r,$ with $d(\alpha,\beta)=r-2d.$ The number $d$ quantifies the \emph{defect} of the path from  being a geodesic. The number of paths with fixed defect  (without the condition of transitivity and monotonicity) were considered in \cite{SchurWeylLevy}.
\end{rmk}
An equivalent representation of the LUE cumulants $\mathcal{C}_W(\mu)$ is the following.
    \begin{prop}\label{Prop:MomentWishart} For $c\ge 1, n\in \N^*, $ and $\mu\vdash n,$
  \begin{align}
  \mathcal{C}_W(\mu)= \sum_{\nu\vdash n, g\ge 0} N^{-2g} c^{n-(2g-2+\#\mu+\#\nu)}  \mathscr{C}_g(\mu,\nu), \label{proprq:CumulantWishartFacto}
  \end{align}
where 
$\mathscr{C}_g(\mu,\nu)$ denotes the number of pairs $(\alpha,\beta)\in S_n^2,$ such that 
\begin{enumerate}
\item $[\alpha]= \mu$ and $[\alpha.\beta]=\nu$
\item  $\#\mu+\#\beta+\#\nu-n=2-2g $
\item the group generated by $\alpha$ and $\beta$ acts transitively on $[n].$
\end{enumerate}  
 \end{prop}
The triple  $(\alpha,\beta,(\alpha.\beta)^{-1})$ is called a \emph{constellation}  of genus $g$, see~\cite[Section 1.2.4]{LandoZvonkine}.  
\par
When $\ell=1$ and $N\to \infty,$ Theorem \ref{MainTh} allows to prove the following duality. 
\begin{cor}\label{cor:duality}  For $c>1$,
\be
\lim_{N\to \infty}\frac{\EE \tr W^{-(n+1)}}{(c-1)^{-(n+1)}} = \lim_{N\to \infty}\frac{\EE\tr W^n}{(c-1)^n} .
\label{eq:duality}
\ee
\end{cor}
This result can be obtained using analytic methods \cite{C,FW,CMOS}. We give here a combinatorial proof relying on a relation between monotone and strictly monotone Hurwitz paths.
\par
The duality~\eqref{eq:duality} is the projection to leading order in $1/N$ of an exact reciprocity law for the LUE recently found in~\cite[Proposition 2.1]{CMOS}:
\be
\EE \tr \left(N W\right)^{-(n+1)}=\left(\prod_{j=-n}^n\frac{1}{\alpha+j}\right)\EE \tr \left(N W\right)^{n}.
\label{eq:reciprocity_or}
\ee 
In the notation of this paper the above relation reads
\be
N^{-(n+1)}\frac{\Gamma(\alpha+n+1)}{\Gamma(n+1)}\mathcal{C}_{W^{-1}}((n+1))=N^{n}\frac{\Gamma(\alpha-n)}{\Gamma(n)}\mathcal{C}_{W}((n)).
\label{eq:reciprocity}
\ee
\par
By Theorem~\ref{MainTh}, it is possible to rephrase the duality~\eqref{eq:reciprocity_or} (or \eqref{eq:reciprocity}) as a functional relation between generating functions of monotone and strictly monotone Hurwitz numbers. Define the formal power series
\begin{align}
\Hm_g(n;x)&=\sum_{\nu\vdash n}x^{-\#\nu}\Hm_g((n),\nu),\label{eq:gen_Hur1}\\
\Hsm_g(n;x)&=\sum_{\nu\vdash n}x^{-\#\nu}\Hsm_g((n),\nu). \label{eq:gen_Hur2}
\end{align}
Then, combining the duality~\eqref{eq:reciprocity_or} with the explicit formulae~\eqref{mainth:CumulantInverse}-\eqref{propeq:CumulantWishartE} for $\mathcal{C}_{W^{-1}}$ and  $\mathcal{C}_{W}$, and comparing the coefficients of the $1/N$-expansions we can get a functional relation for the generating functions~\eqref{eq:gen_Hur1} and~\eqref{eq:gen_Hur2}.
Note that 
\be
\prod_{j=-n}^n\frac{1}{\alpha+j}=\frac{1}{\alpha^{2n+1}}\prod_{j=1}^n\left(1-\frac{j}{\alpha}^2\right)
=\sum_{g\geq0}h_g(1^2,\dots,n^2)\alpha^{-g}
\ee
where 
\ben
h_g(1^2,\dots,n^2)=\sum_{\substack{\ell_1,\cdots,\ell_n\geq0\\\ell_1+\cdots+\ell_n=g}}1^{2\ell_1}2^{2\ell_2}\cdots n^{2\ell_n}
\een
 is the complete symmetric function of degree $g$ evaluated on the square integers $1^2,\dots,n^2$ (see Lemma~\ref{lem:symm} below).
 We learned from~\cite{MN2} that the numbers 
 \be
 T(n+g,n)=h_g(1^2,\dots,n^2)
 \ee
 are known as \emph{Carlitz-Riordan central factorial numbers}, and are given by the explicit formula
 \be
 T(a,b)=2\sum_{j=0}^n(-1)^{b-j}\frac{j^{2a}}{(b-j)!(b+j)!}.
 \ee
 Putting all together we get the following functional equation.
\begin{prop}
\be
\left(\frac{x-1}{x}\right)^{n+1}\Hm_g(n+1;x-1) =n\sum_{j=0}^g\left(\frac{x-1}{x}\right)^{2j}T(n+g-j,n)\Hsm_j(n;x),
\label{eq:func_rel}
\ee
\end{prop}
Functional relations and some explicit formulae for the generating functions of (monotone) Hurwitz numbers have been considered in the literature, see~\cite{GGPNHCIZ,GOULDEN20131,Dubrovin}. To our knowledge, the relation~\eqref{eq:func_rel} is new. It would be interesting to find a combinatorial proof of it.
\par
There exists a duality similar to~\eqref{eq:duality}, for covariances  ($\ell=2$) of LUE moments at leading order in $1/N$. If $\mu=(\mu_1,\mu_2)\vdash n$, then~\cite[Theorem 7.3]{CMOS} 
\be
\lim_{N\to \infty}\frac{\mathcal{C}_{W^{-1}}(\mu)}{(c-1)^{-|\mu|}}=  \lim_{N\to \infty}\frac{\mathcal{C}_{W}(\mu)}{(c-1)^{|\mu|}}.
\label{eq:duality_cov}
\ee
By Theorem~\ref{MainThV2}, this is equivalent to a relation between generating functions
\be
\sum_{r\geq 0}z^r \#\Fm_{n,r,0}(\alpha)=  \sum_{r\geq 0}(z+1)^{n-r} \#\Fsm_{n,r,0}(\alpha)\quad \text{for $z>0$},
\ee
when $\alpha\in S_{n}$ has two cycles $\#[\alpha]=2$.
\par
The enumerative properties of the integer moments of Wishart matrices, suggest to reinterpret various known results in random matrix theory from a combinatorial point of view.  It is known that the moments of LUE (and any other $\beta$-ensemble) satisfy a set of recursions known as `loop equations' (see~\cite[Lemma 7.1]{CMOS}) and it is natural to expect that they have a combinatorial explanations. 
\par
A special property of the LUE, is its connection to the Laguerre polynomials which led Haagerup and Thorbj\o rnsen to discover an exact three-term recursive relation~\cite[Theorem 8.2]{HT} for moments of $W$ (the analogue of the Harer-Zagier recursion of the GUE). Later, it was observed in~\cite{CMSV2} that the Haagerup-Thorbj\o rnsen recursion extends to the moments of $W^{-1}$. For the inverse LUE with parameter $c=2$, the recursion reads
\ben
(N^2-n^2)(n+1)\EE \tr W^{-(n+1)}-3N^2(2n-1)\EE \tr W^{-n}\\
+N^2(n-2)\EE \tr W^{-(n-1)}=0.
\een
Denote by $S(n,d)=\sum_{r\geq0}\#\Fm_{n,r,d}((1\ldots n))$ the number of monotone paths in the Cayley graph on $S_n$ that start at the full cycle $(1\ldots n)$ and, after an  arbitrary (finite) number of steps $\tau_1,\ldots,\tau_r$ have a defect $2d$.
Then, Theorem~\ref{MainThV2} combined with the three-term recursion above gives a recurrence for the numbers $S(n,d)$:
\begin{multline}
(n+1)S(n+1,d+1)-3(2n-1)S(n,d+1)\\
+(n-2)S(n-1,d+1)=n^2(n+1)S(n+1,d)
\label{eq:rec_kg}
\end{multline}
The above recursion appeared in the random matrix approach to the time-delay~\cite[Corollary 1.4]{CMSV2} where the initial conditions are
\be
S(n,0)=\setlength\arraycolsep{1pt}
\; {}_2 F_1\left(\begin{matrix}1-n,n\\2\end{matrix};-1\right),\quad S(0,d)=\delta_{0,d},\quad S(1,d)=\delta_{0,d}.
\ee
Note that $S(n,0)$ is the large Schr\"oder number.
\par
The existing proofs of~\eqref{eq:reciprocity}-\eqref{eq:duality_cov} and~\eqref{eq:rec_kg} are based on special properties of the Laguerre polynomials, but it should be possible to prove these remarkable formulae using algebraic methods. Further study is in progress.

\section{Proofs}
 
\subsection{Proof of the main Theorem}

 We shall give a proof that hinges on the following two propositions. The first one is a restatement of~\cite[Theorems 1 and 4]{GLM} and \cite[Theorems 3.1 and 4.3]{CMS} in a notation which is shorter and better adapted to the purposes of this paper.

 \begin{prop}[\!\!\cite{GLM,CMS}]  \label{prop:MomentCoeff}For any $i,j\in[N]^n$
\ben
\EE\prod_{k=1}^n W_{i(k)j(k) }= N^{-n}  \sum_{\substack{ \sigma\in S_n: \\  i\circ \sigma= j}} \Omega_{n,cN}(\sigma) ,
\een
and,  for $c> 1+\frac{n}{N}$,
\ben
\EE\prod_{k=1}^n W^{-1}_{i(k)j(k) }=(-N)^n  \sum_{\substack{ \sigma\in S_n: \\ i\circ \sigma= j}}   \Omega^{-1}_{n,(1-c)N}(\sigma),
\een
 where for any permutation $\sigma\in S_n$ and $z\in \C,$ 
 $$\Omega_{n,z}(\sigma)= z^{\#\sigma},$$
 whereas  for $|z|>n-1,$ $\Omega^{-1}_{n,z}: S_n\to \C,$ denotes the unique  function   such that 
 $$\Omega^{-1}_{n,z}* \Omega_{n,z}= \Omega_{n,z}*\Omega_{n,z}^{-1}= \delta_{\mathrm{id}}, $$
 where $*$ is the convolution product of functions on the symmetric group $S_n.$   
 \end{prop}
 
The function $\Omega_{n,z}^{-1},$ more commonly denoted by $\mathrm{Wg}_{n,z}$, is called the \emph{unitary Weingarten  function} and admits a remarkable  factorisation property (Proposition~\ref{prop:FactoWeingarten} below). To state it,  we  shall   identify the unital  algebra $(\C^{S_n},*, \delta_{\mathrm{id}})$ with  the group 
algebra $(\C[S_n],., \mathrm{id}),$ that is, the algebra of formal linear combinations of permutations with a product rule extending linearly the product 
of the group $S_n,$ thanks to the isomorphism that maps a function   $f\in \C^{S_n}$ to $\sum_{\sigma\in S_n} f(\sigma)\sigma\in \C[S_n].$  We shall 
keep abusively the same notations for $\Omega_{n,z}$ and $\Omega_{n,z}^{-1}$ viewed as elements of $\C[S_n]$ instead of functions.
\par
The \emph{Jucys-Murphy element} $J_i$ \cite{Jucys} in $\C[S_n]$ is the sum of all transpositions interchanging $i$ with a smaller number:
\begin{align*}
J_1&=0\\
J_2&=(1\,2)\\
J_3&=(1\,3)+(2\,3)\\
&\vdots\\
J_n&= (1\, n)+(2\,n)+\ldots+ (n-1\,n).
\end{align*}
They form a commutative family in the group algebra $\C[S_n]$.

\begin{prop}[\!\!\cite{Collins03,Novak}]  \label{prop:FactoWeingarten}For any $z\in \C,$  

\ben
\Omega_{n,z}= (z+J_1)(z+J_2)\cdots(z+J_n)
\een
 and, for any $z\in \C\setminus \{1-n,2-n,\ldots, n-2,n-1\},$  
  \ben
   \Omega^{-1}_{n,z}= (z+J_1)^{-1}(z+J_2)^{-1}\cdots(z+J_n)^{-1}.
   \een
 \end{prop}
In the proof of  the main theorem we will use classical manipulation of cumulants.
\begin{notation}Setting for partitions $\pi,\nu \in 
\mathcal{P}(n),$ $\pi \le \nu,$    whenever all blocks  of $\pi$ are included in those of $\nu,$ (the partition $\nu$ is said coarser than $\pi$) defines a 
structure of poset on $\mathcal{P}(n)$ with maximal element $1_n= \{[n]\} $ 
and minimal element  $0_n=\{\{1\},\{2\},\ldots ,\{n\}\}.$  For any $\mu \in \mathcal{P}(n),$ we shall write $\mathcal{P}(n)_{\ge \mu} =\{ \pi\in \mathcal{P}
(n): \pi \ge \mu\}$ the set of partitions coarser than $\mu.$
\end{notation}

\begin{lem}[\!\!\cite{Rota}]
\label{lem:MomentCumulant}
Let $\mu \in \mathcal{P}(n)$ be a fixed set partition. For any function $E\colon \mathcal{P}(n)_{\ge \mu}\to \C$, there exists a 
unique  $C\colon \mathcal{P}(n)_{\ge \mu}\to \C$ such that  for all $\pi\in \mathcal{P}(n),$
\begin{equation}
E(\pi)=\sum_{\mu\le \nu \le \pi }    C(\nu).  \label{MomentCumulant}
\end{equation}
\end{lem}
\begin{notation}
If $(Y_1,\ldots, Y_n)$ are $n$ variables on the same probability space,  with all their joint moments of degree less than $n$ and for any  $\pi \in 
\mathcal{P}(n),$ $\EE_\pi(Y_1,\ldots,Y_n)= \prod_{B\in \pi } \EE\prod_{k\in B} Y_k$,  then  the value at a partition 
$\pi\in \mathcal{P}(n)$ of the unique solution to \eqref{MomentCumulant},  is denoted  by $C_{\mu,\pi}(Y_1,\ldots,Y_n)$.  It  is   a \emph{relative 
cumulant}:  for any $n\ge 1, $  $C_{0_n,1_n}(Y_1,\ldots ,Y_n)$, is the cumulant 
$C_n(Y_1,\ldots, Y_n)$, whereas for any $\mu,\pi\in \mathcal{P}(n),$ with $\mu\le \pi,$ 
\ben
C_{\mu,\pi}(Y_1,\ldots, Y_n)= \prod_{S\in \pi} C_{\#\{B\in \mu\colon B\subset S\}} \left( \prod_{k\in B} Y_k, B\in \mu \text{ with } B\subset S \right).
\een
For any pair of transpositions $\tau_1=(a_1\,b_1)$, $\tau_2=(a_2\, b_2),$ with $a_i<b_i$, let 
us write $\tau_1\le \tau_2$ when $b_1\le b_2$. $\Wm_r$ is the set of tuples of transpositions $(\tau_1,\ldots, \tau_r)$ with    $\tau_1\le \tau_2\le \ldots\le \tau_r$. For any partition $\pi \in \mathcal{P}(n),$ let us  denote   by $S_\pi$ the subgroup of $S_n$ consisting of permutations $\sigma\in 
S_n$ with $\sigma(B)=B$ for all blocks $B\in \pi$,  set   $\Wm_r(\pi) = \Wm_r \cap
S_\pi^r$ and for any $A\subset [n],$  $S_A,$  the group of permutations of $A$.
\end{notation}
In the proof we will use the following standard fact on symmetric functions.
\begin{lem}
\label{lem:symm}
For each integer $n\in\N$, and indeterminates $t,x_1,x_2,\dots,x_n$,
\begin{align*}
\prod_{i\geq 1}(1+x_it)=\sum_{r\geq0}e_r(x)t^r,\qquad 
\prod_{i\geq 1}(1-x_it)^{-1}=\sum_{r\geq0}h_r(x)t^r,
\end{align*}
where $e_r(x)=\sum_{i_1<i_2<\cdots<i_r}x_{i_1}x_{i_2}\cdots x_{i_r}$ and $h_r(x)=\sum_{i_1\leq i_2\leq\cdots\leq i_r}x_{i_1}x_{i_2}\cdots x_{i_r}$ are the elementary and complete symmetric functions, respectively.
\end{lem}
 \begin{proof}[Proof of  Theorem \ref{MainThV2}]  Let $\mu=(\mu_1,\ldots, \mu_l) \vdash n$ and $\alpha\in S_n$ be a permutation of type $\mu$ and  
let $\pi_\alpha\in \mathcal{P}(n)$ be the set partitions with blocks given by cycles of $\alpha.$   Multilinearity of cumulants yields
\begin{equation}
C_l(\Tr W^{- \mu_1}, \ldots, \Tr W^{- \mu_\ell})=\sum_{\substack{i,j\in[N]^n\colon\\  i\circ\alpha= j}}  C_{\pi_\alpha,1_n}( W^{-1}_{i(1) j(1)}, \ldots ,W^{-1}_{i(n) j(n)}).
\label{iproof:CumulantTraces}
\end{equation}
According to Proposition \ref{prop:FactoWeingarten}, if $(c-1)N>n$, 
\begin{multline}
(-N)^n\Omega^{-1}_{n,(1-c)N}= \prod_{i=1}^n(c-1- N^{-1}J_i)^{-1} \\=(c-1)^{-n} \sum_{r\ge 0}  h_r(J)  ((c-1)N)^{-r} 
=(c-1)^{-n} \sum_{r\ge 0}    ((c-1)N)^{-r}   \sum_{(\tau_i)_{i=1}^r\in \Wm_r} \tau_1\tau_2\ldots \tau_r.,
\label{WeingartenMonotone}
\end{multline} 
where we used Lemma~\ref{lem:symm} and the fact that the transpositions in $J_i$ are all majorized by the transpositions in $J_j$ when $i<j$.
Combined with Proposition \ref{prop:MomentCoeff}, this leads   for any  $i,j\in[N]^n$ to
$$\begin{aligned}\EE W^{-1}_{i(1) j(1)}\ldots &W^{-1}_{i(n) j(n)}=  \\&(c-1)^{-n} \sum_{r\ge 0} ((c-1)N)^{-r} \#\{ (\tau_i)\in \Wm_r\colon j\circ \tau_1\ldots 
\tau_r=i \}.  
\end{aligned} $$
On the one hand, after relabelling, the same argument applied to each block of a partition $
\pi\in\mathcal{P}(n)$ gives  
\begin{multline*}
\EE_\pi \left(W^{-1}_{i(1) j(1)}, \ldots ,W^{-1}_{i(n) j(n)}\right)=  \\ \prod_{B\in \pi}\left((c-1)^{-\#B} \sum_{r\ge 0} ((c-1)N)^{-r} \#\{ (\tau_i)\in \Wm_r(B)\colon 
j_{|B}\circ \tau_1\ldots \tau_r=i_{|B} \}\right).
\end{multline*} 
Distributing the terms in the  product reads 
$$
(c-1)^{-n}\sum_{(r_{B})_{B\in \pi}\in \N_+^\pi}    \prod_{B\in \pi} ((c-1)N)^{-r_B}\#\{ (\tau_i)\in  \Wm_{r_B}(B)\colon j_{|B}\circ \tau_1\ldots \tau_r=i_{|B} \}.
$$  
Now, for  any ${(r_{B})_{B\in \pi}\in \N_+^\pi}$,   because of  the condition of monotonicity,  for any   collection $(w_B)_{B\in \pi}\in\prod_{B\in \pi}  
\Wm_{r_B}(B),$ there is a unique element of $\Wm_r(\pi)$ whose restrictions to blocks of  $\pi$ is given by $w$, where $r=\sum_{B\in \pi} r_{B}.$  
Hence, $\Wm_r(\pi)$ is in bijection with $\sqcup_{(r_{B})_{B\in \pi}\in \N_+^\pi\colon r=\sum_{B\in \pi} r_{B}} \Wm_{r_B}(B).$ It follows that   the latter expression reads 
\begin{equation}
(c-1)^{-n}\sum_{r\ge 0}    ((c-1)N)^{-r}  \#\{ (\tau_i)\in \Wm_r(\pi): j\circ \tau_1\ldots \tau_r=i\}. \label{iproof:MomentMonotone}
\end{equation}
On the other hand, for any tuple $\mathcal{C}=(\sigma_1,\ldots, \sigma_k)\in S_n^k,$    let $\pi_\mathcal{C}\in\mathcal{P}(n) $ be the set partition given by the orbits of the group $ \langle \sigma_1,\ldots, \sigma_k\rangle $   
and set for any $\pi\ge\nu \ge \pi_\alpha, r\ge 1,$  
$$\Wm_r(\nu, \pi)= \{(\tau_i)_i\in \Wm_r(\pi): \pi_{\alpha, \tau_1,\ldots, \tau_r}=\nu  \}.$$
Then, \eqref{iproof:MomentMonotone}  implies that for any $\pi \in \mathcal{P}(n)_{\ge \pi_\alpha},$ 
\begin{align*}
&\EE_\pi \left(W^{-1}_{i(1) j(1)}, \ldots ,W^{-1}_{i(n) j(n)}\right)=\\
&\sum_{\pi_\alpha\le\nu \le \pi } (c-1)^{-n}\sum_{r\ge 0}    ((c-1)N)^{-r}  \#\{ (\tau_i)\in \Wm_r(\nu,\pi)\colon j\circ \tau_1\ldots \tau_r=i\}.
\end{align*}
Using Lemma \ref{lem:MomentCumulant}, it follows that  for all $i,j\in[N]^n$  and $ \nu\ge \pi_\alpha,$
\begin{align}
 C_{\nu, 1_n}(& W^{-1}_{i(1) j(1)}, \ldots ,W^{-1}_{i(n) j(n)})=\nonumber\\
&  (c-1)^{-n} \sum_{r\ge 0}    ((c-1)N)^{-r}  \#\{ (\tau_i)\in \Wm_r(\nu,1_n)\colon j\circ \tau_1\ldots \tau_r=i\}. \label{iproof:CumulantMonotoneI}
\end{align}
With this equation, we can now look back at \eqref{iproof:CumulantTraces}  and write
 \begin{multline*} C_\ell(\Tr W^{-\mu_1} , \ldots, \Tr W^{-\mu_\ell})\\
 = (c-1)^{-n}\sum_{\substack{r\ge 0} } \sum_{\substack{i,j\in[N]^n:\\  i\circ\alpha= j}} ((c-1)N)^{-r}  \#\{ (\tau_i)\in \Wm_r(\pi_\alpha,1_n)\colon j\circ \tau_1\ldots \tau_r=i\}.
\end{multline*}
For any $\beta \in S_\pi$ and  $r\ge 1,$ let us consider $$
\Wm_r(\pi_\alpha, \pi,\beta)=\{ (\tau_i)_{i=1}^r\in \Wm_r(\pi_\alpha,\pi)\colon  \alpha \tau_1\ldots 
\tau_r= \beta \}.
$$
Fixing $r\ge 1$ in the last sum,  the coefficient of $(c-1)^{-n-r}$ is 
\begin{align*}
&N^{-r}\sum_{\substack{ \beta\in S_n,\\ (\tau_i)_{i=1}^r\in \Wm_r(\pi_\alpha, 1_n,\beta)}}    \#\{i,j\in[N]^n\colon j\circ (\alpha^{-1}\beta)=i, i\circ \alpha= j\} \\
&=\sum_{\substack{ \beta\in S_n,\\ (\tau_i)_{i=1}^r\in \Wm_r(\pi_\alpha, 1_n,\beta)}}  N^{ \#\beta-r}.
\end{align*}
Now according to Riemann-Hurwitz formula \cite[Remark 1.2.21]{LandoZvonkine}, for any $\beta\in S_n$, $(\tau_i)_{i=1}^r\in \Wm_r(\pi_\alpha, 1_n,\beta),$  $
\#\alpha+\#\beta-r= 2-2d, $ for some $d\in \N.$ Therefore, the last right-hand-side is
$$\sum_{r\ge 0, d\ge 0 }  N^{2 -2d-\#\alpha}  \#\Fm_{n,r,d}.$$
The first claim~\eqref{eq:1st} follows by inspection. 
\par
The second claim~\eqref{eq:2nd} follows from the very same argument if, instead of~\eqref{WeingartenMonotone}, we start from the expression
\begin{multline}
\label{WeingartenStrictlyMonotone}
(N)^{-n}\Omega_{n,cN}= \prod_{i=1}^n(c+ N^{-1}J_i)=c^{n} \sum_{r\ge 0}  e_r(J)  (cN)^{-r}   \\
 =c^{n} \sum_{r\ge 0}    (cN)^{-r}   \sum_{(\tau_i)_{i=1}^r\in \Wsm_r} \tau_1\tau_2\ldots \tau_r, 
\end{multline} 
where $\Wsm_r$ is the set of strictly monotone tuples of transpositions $(\tau_1,\ldots, \tau_r)$,    $\tau_1< \tau_2< \ldots< \tau_r$. The proof of formula~\eqref{eq:2nd} proceeds mutatis mutandis with $\Wm_r$ replaced by $\Wsm_r$.
 The details of the calculations are left to the Reader.
\end{proof}

We can now easily conclude.

\begin{proof}[Proof of Theorem \ref{MainTh} and Proposition \ref{Prop:MomentWishart}]  Unfolding the definitions of these statements and of Theorem 
\ref{MainThV2}, we get that for any $\mu\vdash n, r,g\ge 0,$
\begin{align*} 
\sum_{\alpha\in S_n: [\alpha]=\mu}\# \Fm_{n,r,d}(\alpha) &= \sum \Hm_{d}(\mu,\nu) \text{ \,\, and  }\sum_{\alpha\in S_n: [\alpha]=\mu}\# \Fsm_{n,r,d}(\alpha) &= \sum \Hsm_{d}(\mu,\nu), 
\end{align*}
where in the right-hand-sides, we sum over ${\nu \vdash n}$ with $\#\nu= \#\mu+ r-2d.$ The claims of Theorem \ref{MainTh} follow by inspection.
\par
To prove Proposition \ref{Prop:MomentWishart}, let us  recall that any permutation $\sigma\in S_n$ can be uniquely factorized as $\sigma=\tau_1\ldots \tau_{|\sigma|}$ where  $(\tau_1,\ldots,\tau_{|\sigma|})$ is a strictly monotone tuple and $|\sigma|=n-\#\sigma$. Moreover, for any $\alpha\in S_n,$ $(\alpha,\sigma )$ acts transitively on $[n]$ if and only if $(\alpha,\tau_1,\ldots, \tau_{|\sigma|})$ does. Hence  considering  for any constellation  $(\alpha,\beta,(\alpha\beta)^{-1}),$ the unique tuple $(\tau_1,\ldots, \tau_{|\beta|})$ with $\tau_1<\tau_2<\ldots <\tau_{|\beta|}$ such that $\tau_1\ldots \tau_{|\beta|}=\beta$  leads to    $\mathscr{C}_g(\mu,\nu)=\Hsm_{g}(\mu,\nu),$ for any $\nu,\mu\vdash n$ and $g\ge 0.$  
\end{proof}

\begin{rmk} Proposition \ref{Prop:MomentWishart} can be proved more directly along the lines of the proof of Theorem \ref{MainThV2} starting from  the expression $(N)^{-n}\Omega_{n,cN}=\sum_{\beta\in S_n} N^{\#\beta-n} c^{\#\beta},$ without factorizing into transpositions.
\end{rmk}
\begin{rmk}
Let us emphasize that in the proof of Theorem \ref{MainThV2}, the monotonicity condition  was crucial for a factorisation property of the set of  
partitioned monotone paths to get \eqref{iproof:MomentMonotone}.    
\end{rmk}
\begin{rmk}
Proposition \ref{prop:MomentCoeff} can be read as an equality of tensors in $\mathrm{End}((\C^N)^{\ts n})$. The left-hand-side 
commutes with the diagonal action of unitary matrices, whereas the right-hand-side can be viewed as the endomorphism given by the linear 
combination of permutations of tensors. (As already mentioned, this is an instance of \emph{Schur-Weyl duality}.)  It would have been more elegant but less elementary to 
write the above proof in this language. 
\end{rmk}
 \subsection{A combinatorial proof of a duality formula}
The two formulae in Theorem \ref{MainTh} have a striking similarity that we shall use  to deduce 
Corollary \ref{cor:duality}.   Therefor,  we shall use the following  decompositions of monotone  minimal factorisations of a full cycle.

 Denoting by  $\mathrm{T}_n$   the set of all transpositions  of $S_n$, we consider\footnote{We borrow here some notations from \cite{LFacto} but  
do not develop the relation with parking functions which would deserve further consideration.  } for $r\ge 0,$ 
$$\Fm_{n,r}= \{ (\tau_1,\ldots, \tau_r )\in \Ton_n^r :  \# (1 \,2 \ldots\, n) \tau_1\ldots \tau_r= r+1, \tau_1\le \tau_2\le \ldots \le \tau_r  \}$$
and 
$$\Fsm_{n,r}=\{ (\tau_1,\ldots, \tau_r )\in \Fm_r: \tau_1 <\tau_2<\ldots < \tau_r  \},$$ 
where by convention  the empty sequence is the only element of $\Fm_{n,0}=\Fsm_{n,0}=\{(\empty)\}$. We wish to relate the  sets 
  $\Fm_n=\cup_{r\ge 1} \Fm_{n,r}$   and $\Fsm_{n}=\cup_{r\ge 0} \Fsm_{n,r}.$ 
Let us   define a map
$$\Phi_n: \Fm_{n+1}\longrightarrow \Fsm_n$$
setting  for all $w=(\tau_1,\ldots, \tau_r)\in\Fm_{n+1,r}$ given by  $((a_1 \,b_1),\ldots, (a_r \,b_r))$, with $a_i<b_i$ for all $i$, 
$$ \Phi(w)= (\tau_{i_1},\ldots, \tau_{i_l}),$$
where  $(i_1,\ldots, i_l)$ are the record times of the sequence $(b_1,\ldots, b_r),$  before reaching $n+1,$ defined inductively as follows. If $b_{1}=n+1,$ $l=0$ and   $\Phi_n(w)= ().$ If $b_1\le n,$ $i_1=1$ and   $i_{m+1}=\inf\{t> i_{m}: b_t>b_{i_m} \}$ as long as $b_{i_{m+1}}\le n$,   while we 
set $l=m$  when $b_{i_{m+1}}>n$.   For instance,   $\Phi_4( (1\,3) (2\,3) (1\,5)(4\,5)  )= ( (1\,3)).$  The main observation to prove 
the duality of Corollary \ref{cor:duality} can be stated as follows.

\begin{lem} \label{lem:Preimage}For any $l\ge 0,$ $w\in \Fsm_{n,l}$  and $r\ge l,$
$$ \#\Phi_n^{-1}(w)\cap \Fm_{n+1,r}={n-l\choose r-l}. $$
\end{lem}   

\begin{proof}  Let us recall that for any permutation $ \sigma\in S_{n}$ and any transposition $(a\, b),$ $ \#\sigma . ( a \,b) -\# \sigma$ is whether $1,$ 
when $ a$ and $b$ are in the same orbit of $\sigma,$  or $-1$ otherwise.  From this geometric fact follow two observations.  When  $(\tau_1,\ldots, 
\tau_r)\in \Fm_{n,r},$  

\begin{itemize}
\item[1.] for all $m\le r,$  $\# (1\, 2\ldots \, n) \tau_1\ldots \tau_m= m+1$; 
\item[2.] for all $m\le r-1,$  writing $\tau_m=(a\, b)$ and $\tau_{m+1}=(c\, d),$ with $a<b$ and $c<d,$ then 
\begin{itemize}
\item[$\bullet$] whether $[c,d]\supset [a,b),$  
\item[$\bullet$]  or $d=b$ and $c>a.$
\end{itemize} 
\end{itemize}
Hence,  any sequence $(\tau_1,\ldots, \tau_r)\in \mathcal{F}_{n+1,r}$ can be written uniquely as 
\begin{align} \label{iproofDuality:Decomposition}(a_1\, b_1 ) ,(a_2\, b_1),\ldots ,&(a_{i_2-1}\, b_1) ,(a_{i_2}\, b_2) ,\ldots  ,(a_{i_3-1}\, b_2) ,\ldots  \\ & 
\ldots, (a_{i_l}\, b_l ),\ldots, (a_{i_{l+1}-1}\, b_l ), (a_{i_{l+1}} \, n+1)\ldots,  (a_r\, n+1), \nonumber \end{align}
where $1\le l\le r,$ $b_1<b_2<\ldots< b_l,$   $1= i_1<i_2<\ldots < i_l< i_{l+1} \le r+1 $  and  for any $m \in [l+1],$ 
\ben
\text{$ a_{i_{m}}<a_{i_{m}+1}<\ldots <a_{i_{m+1}-1}$ with $\{a_{i_{m}},\ldots, a_{i_{m+1}-1}\} \cap (a_{i_{j}} \, b_j]= \emptyset, $ for all 
$j<m,$}
\een
or as 
\begin{equation}
\label{iproofDuality:DecompositionEmpty}(a_1\, n+1 ) ,(a_2\, n+1),\ldots ,(a_{r}\, n+1),
\end{equation}
with  $ 1\le a_{1}<a_{2}<\ldots <a_{r}\le n$.  When $i_{l+1}=r+1,$ by convention, no transposition acts on $n+1.$  As illustrated in Figure \ref{iproof:FigDuality}, it follows that for any  $0\le l<n$ and $w= ((x_1 \, b_1),
\ldots , (x_l \, b_l))\in \Fsm_{n,l},$  with $x_i<b_i$ for all $i\in [l],$ the map 
$$\Psi:  \Phi_n^{-1}(w)  \longrightarrow\{S\in \mathcal{P}([n]): S \subset [n]\setminus \{b_1,\ldots,b_l \}\} $$
that maps a sequence decomposed as in \eqref{iproofDuality:Decomposition} or \eqref{iproofDuality:DecompositionEmpty} to $\{a_1,\ldots, a_r\} 
\setminus \{a_{i_1},\ldots, a_{i_l}\}$  and resp. $\{a_1,\ldots, a_r\}$  when $l=0$, is a bijection such that $\Psi(\Fm_{n+1,r} \cap \Phi_n^{-1}(w))=\{S\in 
\mathcal{P}([n]\setminus \{b_1,\ldots, b_l\}): \#S=r-l \}. $ The claim follows.
\begin{figure}
\includegraphics[width=.75\columnwidth]{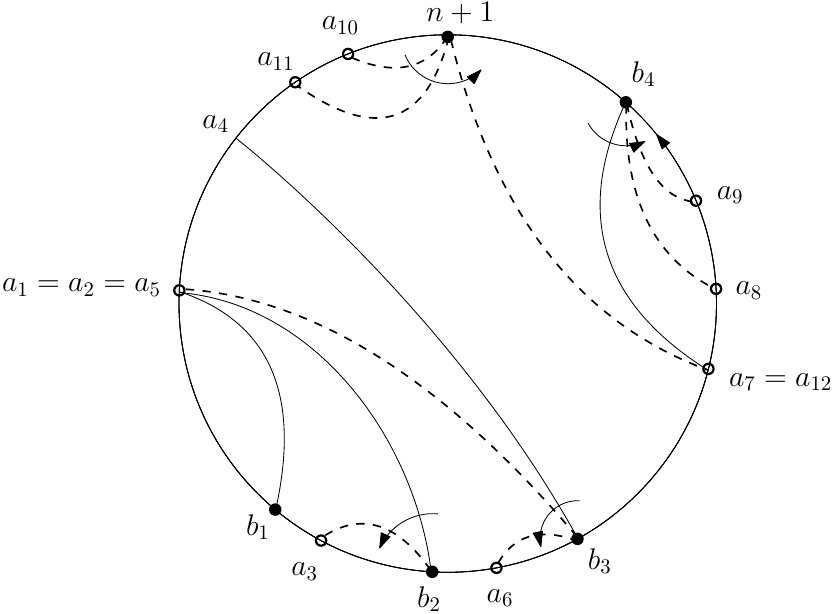}
\caption{ \label{iproof:FigDuality} Representation of the decomposition of an element $w\in\Fm_{n+1,12}$, where each transposition is 
represented by a strand that is dotted when it does not belong to $\Phi_n(w)\in \Fsm_{n,4}.$ The set of white dots is  $\Psi(w)$. The order of 
composition of the transpositions knowing only $\Psi(w)$  and the set of black dots is given first by the counter-clockwise order of the black dots and 
then by the counter-clockwise order of white dots around each black dot.}
\end{figure}
\end{proof}

 \begin{proof}[Proof of Corollary \ref{cor:duality}]   Thanks to Theorem \ref{MainThV2}, applied to $\sigma= (1\ldots n+1),$ 
  \begin{align}
 \lim_{N\to \infty} (c-1)^{2n+1} \EE\tr W^{-n-1} = \sum_{r=0}^n (c-1)^{n-r} \#\Fm_{n+1,r}
 \end{align}
 and applied to $\sigma=(1\ldots n)$, 
 \begin{equation}
  \lim_{N\to \infty}  \EE\tr W^{n}= \sum_{r =0}^{n-1} c^{n-r}  \#\Fsm_{n,r}.
 \end{equation}
 But applying Lemma \ref{lem:Preimage} gives 
\begin{multline*}
  \sum_{r=0}^n (c-1)^{n-r} \#\Fm_{n+1,r}=   \sum_{l=0}^n  \sum_{w \in \Fsm_{n,l}} \sum_{r=l}^n (c-1)^{n-r}\# \Phi_n^{-1}(w)  \cap \Fm_{n+1,r} \\
  =\sum_{l=0}^n  \sum_{w \in \Fsm_{n,l}} \sum_{r=l}^n (c-1)^{n-r} {n-l \choose r-l} = \sum_{l=0}^{n-1} c^{n-l}  \#\Fsm_{n,l}.
 \end{multline*} 
 \end{proof}
 
 \section*{Acknowledgements}
 The authors would like to thank Gregory Berkolaiko, Francesco Mezzadri, James Mingo and Marcel Novaes for feedback and helpful remarks on the first draft of this work.
 \bibliographystyle{abbrv}
\bibliography{biblioMH2}

  \end{document}